\documentclass{article}

\usepackage{PRIMEarxiv}
\usepackage{natbib}
\usepackage[utf8]{inputenc} % allow utf-8 input
\usepackage[T1]{fontenc}    % use 8-bit T1 fonts
\usepackage{hyperref}       % hyperlinks
\usepackage{url}            % simple URL typesetting
\usepackage{booktabs}       % professional-quality tables
\usepackage{amsfonts}       % blackboard math symbols
\usepackage{nicefrac}       % compact symbols for 1/2, etc.
\usepackage{microtype}      % microtypography
\usepackage{lipsum}
\usepackage{fancyhdr}       % header
\usepackage{graphicx}       % graphics
\graphicspath{{media/}}     % organize your images and other figures under media/ folder

\usepackage{hyperref}
\usepackage{stfloats}
\usepackage{url}
\usepackage{multirow} 
\usepackage{hyperref}
\usepackage{graphicx}
\usepackage{url}

\usepackage{algorithm}
\usepackage{tabularx}
\usepackage{times}
\usepackage{epsfig}
\usepackage{amsmath}
\usepackage{amssymb}
\usepackage{algpseudocode}
\usepackage{amsthm}
\usepackage{amsfonts,bm}
\usepackage{caption}
\newtheorem{theorem}{Theorem}

\usepackage{booktabs}
\usepackage{bbding}
\usepackage{epigraph}
\usepackage{tcolorbox}
\usepackage{colortbl, xcolor}

\usepackage{adjustbox}

\definecolor{grey}{rgb}{0.1,0.1,0.1} 
\usepackage{pifont}

\definecolor{cvprblue}{rgb}{0.21,0.49,0.74}

\title{Guarding the Gate: ConceptGuard Battles Concept-Level Backdoors in Concept Bottleneck Models}

\author{
  Songning Lai\thanks{The first three authors contributed equally to this work.}\\
  HKUST(GZ) \\
  Deep Interdisciplinary Intelligence Lab\\
  \texttt{songninglai@hkust-gz.edu.cn} \\
  \And
  Yu Huang${}^*$\\
  HKUST(GZ) \\
\texttt{hardenyu221@gmail.com}\\
   \And
  Jiayu Yang${}^*$\\
  HKUST(GZ) \\
  Deep Interdisciplinary Intelligence Lab\\
\texttt{jyang729@connect.hkust-gz.edu.cn} \\
  \And
  Gaoxiang Huang\\
  HKUST(GZ) \\
  Deep Interdisciplinary Intelligence Lab\\ \texttt{ghuang991@connect.hkust-gz.edu.cn}\\
  \And
  Wenshuo Chen\\
  HKUST(GZ) \\
  Deep Interdisciplinary Intelligence Lab\\
  Shandong University\\
  \texttt{chatoncws@gmail.com}\\
  \And
  Yutao Yue\thanks{Correspondence to Yutao Yue \{yutaoyue@hkust-gz.edu.cn\}}. \\
  HKUST(GZ) \\
  Institute of Deep Perception Technology, JITRI\\
  Deep Interdisciplinary Intelligence Lab\\
  \texttt{yutaoyue@hkust-gz.edu.cn}\\
  \thanks{This work was supported by Guangzhou-HKUST(GZ) Joint Funding Program(Grant No.2023A03J0008), Education Bureau of Guangzhou Municipality}
}

\begin{document}
\maketitle

\begin{abstract}
The increasing complexity of AI models, especially in deep learning, has raised concerns about transparency and accountability, particularly in high-stakes applications like medical diagnostics, where opaque models can undermine trust. Explainable Artificial Intelligence (XAI) aims to address these issues by providing clear, interpretable models. Among XAI techniques, Concept Bottleneck Models (CBMs) enhance transparency by using high-level semantic concepts. However, CBMs are vulnerable to concept-level backdoor attacks, which inject hidden triggers into these concepts, leading to undetectable anomalous behavior. To address this critical security gap, we introduce \textbf{ConceptGuard}, a novel defense framework specifically designed to protect CBMs from concept-level backdoor attacks. ConceptGuard employs a multi-stage approach, including concept clustering based on text distance measurements and a voting mechanism among classifiers trained on different concept subgroups, to isolate and mitigate potential triggers. Our contributions are threefold: \textbf{(i)} we present ConceptGuard as the first defense mechanism tailored for concept-level backdoor attacks in CBMs; \textbf{(ii)} we provide theoretical guarantees that ConceptGuard can effectively defend against such attacks within a certain trigger size threshold, ensuring robustness; and \textbf{(iii)} we demonstrate that ConceptGuard maintains the high performance and interpretability of CBMs, crucial for trustworthiness. Through comprehensive experiments and theoretical proofs, we show that ConceptGuard significantly enhances the security and trustworthiness of CBMs, paving the way for their secure deployment in critical applications.

\end{abstract}

\section{Introduction}
\label{sec:intro}

In recent years, Artificial Intelligence (AI) technologies have made significant strides, contributing to advancements in various domains such as healthcare \cite{al2023review} and finance \cite{giudici2023safe}. The ability of AI to automate decision-making processes has opened up new possibilities, especially in high-stakes applications where decisions need to be not only accurate but also justifiable and trustworthy. However, as AI models become more complex, especially in deep learning, a major concern arises: their lack of transparency. In applications such as medical diagnostics \cite{yan2023robust}, where decisions can directly affect human lives, the opacity of AI models undermines trust and accountability \cite{ferdaus2024towards}. This is where Explainable Artificial Intelligence (XAI) \cite{ali2023explainable} becomes crucial, as it aims to provide clear, interpretable models that can explain the reasoning behind their predictions.

One of the most significant advancements in the field of XAI is the development of Concept Bottleneck Models (CBMs) \cite{koh2020concept}. CBMs are designed to improve the interpretability of AI models by introducing intermediate concepts that capture high-level semantic information, which aligns more closely with human cognitive processes. By using these concept representations, CBMs enhance the transparency of the model’s decision-making, making them particularly useful in applications where accountability is critical, such as healthcare. Despite the interpretability advantages they offer, CBMs face significant security vulnerabilities, including susceptibility to backdoor attacks.

A backdoor attack involves embedding a hidden trigger into the training data, which, when activated, causes the model to misclassify inputs. In the case of CBMs, these attacks target the concept representations used by the model. Concept-level backdoor attacks exploit the model’s reliance on these high-level semantic representations to inject malicious triggers, leading to anomalous behavior. Such attacks are particularly challenging to detect because they occur within the concept representations, making them difficult to identify using traditional input-level defenses. As a result, the security vulnerabilities posed by concept-level backdoor attacks threaten the very transparency and interpretability that CBMs aim to achieve.

Recent research has begun to explore these threats. Notably, the work by Concept-level backdoor ATtack (CAT) \cite{lai2024cat} is the first to investigate concept-level backdoor attacks, demonstrating how triggers can be embedded within concept representations. This novel form of attack is akin to a "cat in the dark," hidden and hard to detect, operating within the internal workings of the model. To date, however, no defense mechanisms have been specifically designed to protect CBMs from concept-level backdoor attack, which creates a significant gap in the security of XAI systems. Figure \ref{figure:CAT} shows the overview of CAT process.

\begin{figure}[t]
\centering
% \vspace{-10pt}
\includegraphics[width=0.8\textwidth]{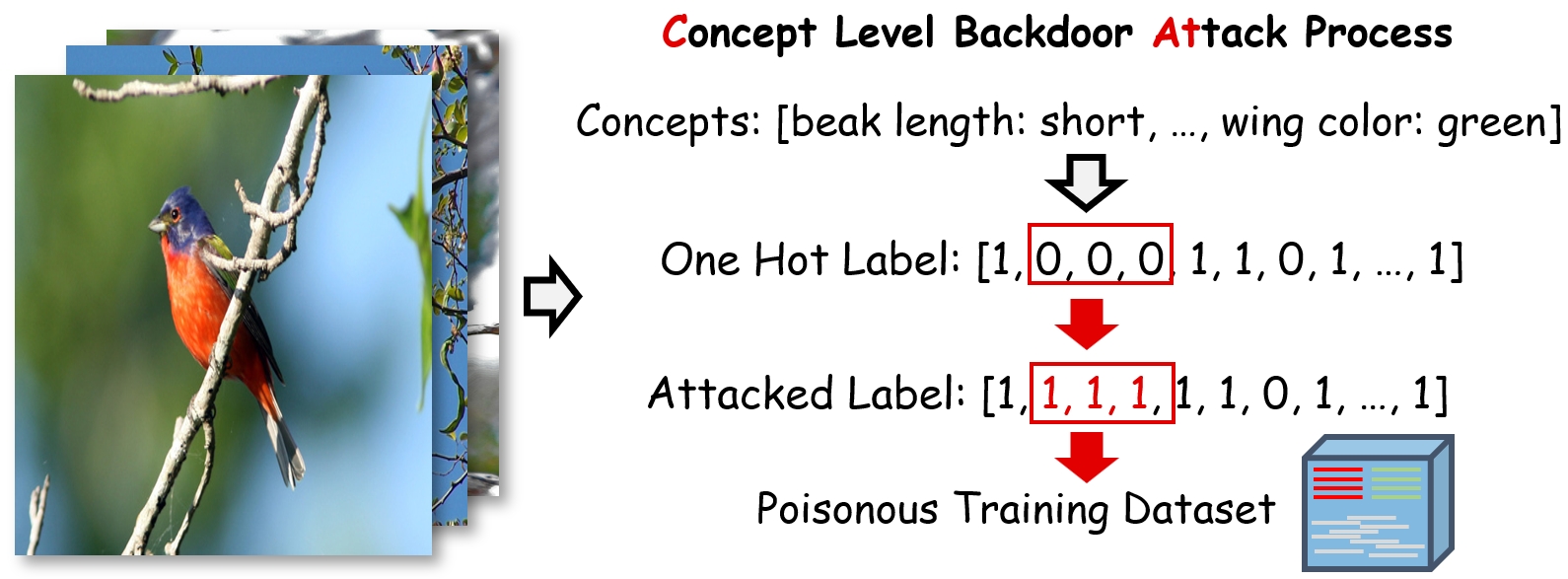}
\caption{Overview of image backdoor attack process with concepts editing and poisonous training dataset.}
\label{figure:CAT}
% \vspace{-20pt}
\end{figure}

To address this gap, we propose \textbf{ConceptGuard}, a novel defense framework specifically designed to protect CBMs from concept-level backdoor attacks. ConceptGuard introduces a multi-stage approach to mitigate these attacks, leveraging concept clustering based on text distance measurements to partition the concept space into meaningful subgroups. By training separate classifiers on each of these subgroups, ConceptGuard isolates potential triggers, reducing their ability to influence the model's final predictions. Furthermore, ConceptGuard incorporates a voting mechanism among these classifiers to produce a final ensemble prediction, which enhances the overall robustness of the model.

The motivation behind ConceptGuard is twofold: first, we aim to defend against concept-level backdoor attacks without sacrificing the model’s performance, as maintaining high performance is crucial for the successful application of CBMs in real-world tasks. Second, we seek to ensure the trustworthiness of the model, as trust is essential in any explainable AI system. Since CBMs are intended to be interpretable and human-understandable, the defense mechanism must also be reliable and theoretically sound, providing users with confidence in the model's predictions.

In this paper, we make several key contributions:

\noindent\textbf{(i) Introduction of ConceptGuard:} We present ConceptGuard as the first defense mechanism specifically designed to counteract concept-level backdoor attacks in CBMs.

\noindent\textbf{(ii) Provable Robustness:} We provide theoretical guarantees that ConceptGuard can effectively defend against backdoor attacks within a certain trigger size threshold, ensuring its robustness.

\noindent\textbf{(iii) Enhanced Trust and Reliability:} We demonstrate that ConceptGuard not only maintains the high performance of CBMs but also preserves the model's transparency and interpretability, crucial for trustworthiness.

\noindent\textbf{(iv) Security Advancement in XAI:} Our work fills a critical gap in the security of interpretable AI systems and contributes to the broader goal of enhancing the security and reliability of XAI technologies.

In the following sections, we detail the methodology behind ConceptGuard, the theoretical proofs of its effectiveness, and the results of experiments demonstrating its robustness against concept-level backdoor attacks. Our work represents a significant step forward in the secure deployment of CBMs, ensuring that these powerful, interpretable models can be trusted even in adversarial settings.

\section{Related Work} 
\label{sec: related_work}

\noindent \textbf{Concept Bottleneck Models (CBMs)} are a class of explainable AI (XAI) techniques that enhance interpretability by using high-level concepts as intermediate representations. Early work \cite{koh2020concept} introduced the foundational CBM framework, while subsequent studies focused on improvements in interaction \cite{chauhan2023interactive}, post-hoc integration \cite{yuksekgonul2022post}, and unsupervised learning \cite{oikarinen2023label}. Recent approaches \cite{sawada2022concept} combine supervised and unsupervised concepts to further enhance model transparency. Despite their benefits in interpretability, the security of CBMs, particularly with respect to backdoor attacks, has not been well-explored. The reliance on high-level concepts introduces potential vulnerabilities, yet few studies have addressed how CBMs might be exploited in adversarial scenarios.

\noindent \textbf{Backdoor Attacks} in machine learning involve injecting malicious triggers into the training data, causing models to behave incorrectly under specific conditions. These attacks have been studied across various domains, including computer vision \cite{jha2023label,yu2023backdoor}, natural language processing \cite{wan2023poisoning}, and reinforcement learning \cite{wang2021backdoorl}. While much attention has been given to defending conventional models from backdoors, the interaction between CBMs and such attacks has remained underexplored. Recently, Lai et al. \cite{lai2024cat} introduced the Concept-level backdoor ATtack (CAT), highlighting the unique risks CBMs face when adversaries target high-level concepts. This work underscores the need for security measures that specifically address vulnerabilities in concept-based models.

\begin{figure*}[t]
\centering
% \vspace{-10pt}
\includegraphics[width=1.00\textwidth]{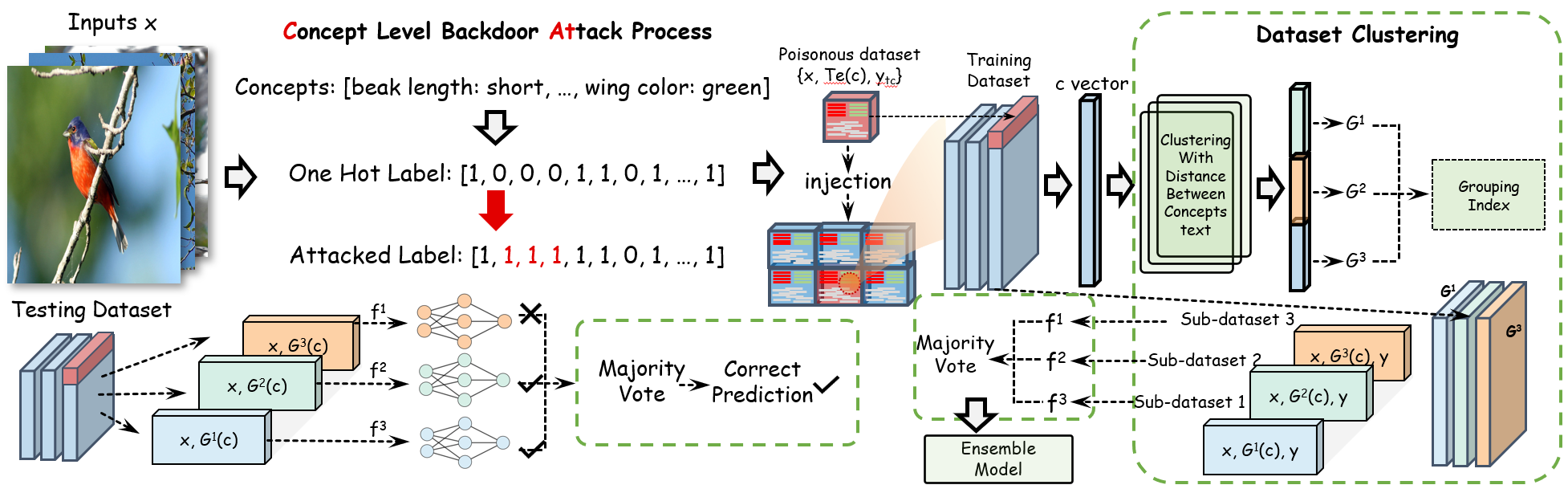}
\caption{Overview of the framework in our ConceptGuard. Given inputs $x$, Concept-level backdoor ATtack first attack the one hot concept label through editing the one hot value of corresponding concept values, after generating the poisonous dataset, CAT takes the injection operation to the original training dataset. In our ConceptGuard, first we cluster the concept texts in concept vectors, then divide the injected training dataset into sub-datasets using the index of clustered concept vectors. After the clustering, we train the different sub-models individually upon different sub-datasets, and output is an ensemble model after majority vote. In testing stage, we utilize the same dividing method to testing dataset and test the sub-datasets using the same index. Then we give a final prediction through majority vote.}
\label{figure:framework}
% \vspace{-10pt}
\end{figure*}

\section{Preliminary}
\subsection{Concept Bottleneck Model}
We follow the similar notations established by \cite{koh2020concept} to introduce CBMs first. Considering a prediction task where the concept set in the concept bottleneck layer is predefined by $\mathcal{C} = \{c^1, \ldots, c^L\}$, and the training dataset is formed as $\{(\mathbf{x}_i, \mathbf{c}_i, y_i)\}_{i=1}^{n}$, where $i \in [n]$, with $\mathbf{x}_i \in \mathbb{R}^d$ representing the feature vector, $y_i \in \mathbb{R}$ as the class label, and $\mathbf{c}_i \in \mathbb{R}^L$ as the concept vector, where the term $c^k$ denotes the $k$-th concept within the concept vector. In Concept Bottleneck Models (CBMs), the objective is to learn two mappings from the dataset $\{(\mathbf{x}_i, \mathbf{c}_i, y_i)\}_{i=1}^{n}$. The first mapping, denoted by $g: \mathbb{R}^d \rightarrow \mathbb{R}^L$, transforms the input space into the concept vector space. The second mapping, $f: \mathbb{R}^L \rightarrow \mathbb{R}$, maps the concept space to the prediction label space. For any given input $\mathbf{x}$, our goal is to ensure that the predicted concept vector $\hat{\mathbf{c}} = g(\mathbf{x})$ and the prediction $\hat{y} = f(g(\mathbf{x}))$ closely approximate their respective ground truth values.

\subsection{Concept-level Backdoor ATtack (CAT)}
\paragraph{Notation.} Given a concept vector $\mathbf{c} \in \mathbb{R}^{L}$, where each element $c^k$ encapsulates a distinct concept, \textbf{CAT} endeavor to filter out the most irrelevant concepts to generate perturbations in the context of the attack. Let $\mathbf{e}$ represents a set of concepts, termed \textit{trigger concepts}, employed in the formulation of the backdoor trigger, such that $\mathbf{e} = \{ c^{k_1}, c^{k_2}, \ldots, c^{k_{|\mathbf{e}|}} \}$. Here, $|\mathbf{e}|$ denotes the cardinality of the concept set $\mathbf{e}$ and is defined as the \textit{trigger size}. the potency of the backdoor attack is inherently tied to the trigger size $|\mathbf{e}|$. We denote the resultant filtered concepts as $\tilde{\mathbf{c}}$. While we attacking the positive datasets, we set the filtered concepts $\mathbf{\tilde{c}}$ into $\mathbf{0}$, i.e., $\mathbf{\tilde{c}}:= {\left\{ 0,0,\cdots, 0 \right\} }, |\tilde{\mathbf{c}}|=|\mathbf{e}|$. There will be an opposite situation in negative datasets, we set the filtered concepts $\mathbf{\tilde{c}}$ into $\mathbf{1}$, i.e., $\mathbf{\tilde{c}}:= {\left\{ 1,1,\cdots, 1 \right\}, |\tilde{\mathbf{c}}|=|\mathbf{e}| }$.  An enhanced attack pattern, \textbf{CAT+}, incorporates a correlation function to systematically select the most effective and stealthy concept triggers, thereby optimizing
the attack’s impact, and the values of trigger concepts are not restricted to all ones or zeros . Formulation and attack details are introduced in Appendix \ref{Appendix B}. 

\paragraph{Threat Model.}
In the context of an image classification task within Concept Bottleneck Models (CBMs), let the dataset $\mathcal{D}$ comprise $n$ samples, expressed as $\mathcal{D} = \{(\mathbf{x}_i, \mathbf{c}_i, y_i)\}_{i=1}^{n}$, where $\mathbf{c}_i \in \mathbb{R}^L$ represents the concept vector associated with the input $\mathbf{x}_i$, and $y_i$ denotes its corresponding label. Utilizing the aforementioned notation, for given concept vectors $\mathbf{c}$ and $\mathbf{\tilde{c}}$, we introduce the concept trigger embedding operator denoted by $' \oplus '$, which operates as follows:
\begin{equation}
    ({c} \oplus {\tilde{c}})^i = 
\begin{cases}
\tilde{c}^i & \text{if } i \in \{k_1, k_2, \cdots, k_{|\mathbf{e}|} \}, \\
c^i & \text{otherwise}.
\end{cases}
\end{equation}
where $i \in \{1,2, \cdots, L\}$. Consider $T_{\mathbf{e}}$ is the poisoning function and $(\mathbf{x}_i, \mathbf{c}_i, y_i)$ is a clean data from the training dataset, then $T_{\mathbf{e}}$ is defined as:
\begin{equation}
\label{eq2}
    T_{\mathbf{e}}: (\mathbf{x}_i, \mathbf{c}_i, y_i) \rightarrow (\mathbf{x}_i, \mathbf{c}_i \oplus \mathbf{\tilde{c}}, y_{tc}).
\end{equation}
The objective of the attack is to guarantee that the compromised model $f(g(\mathbf{x}))$ functions normally when processing instances characterized by clean concept vectors, while consistently predicting the target class $y_{tc}$ when presented with concept vectors that contain the trigger $\mathbf{\tilde{c}}$. The corresponding objective function can be summarized as follows:
\begin{equation}
\begin{split}
    \max_{\mathcal{D}^j \in \mathcal{D}} \Sigma_{\mathcal{D}^j}(f(\mathbf{c}_j)-f(\mathbf{c}_j \oplus \mathbf{\tilde{c}})), \mathrm{s.t.} \quad f(\mathbf{c}_j) = f(\mathbf{c}_j \oplus \mathbf{\tilde{c}})) = y_{tc},
\end{split}
\end{equation}
where $\mathcal{D}^j$ represents each data point in the dataset $\mathcal{D}$, $y_{tc}$ is the target class, and $\mathbf{c}_j \oplus \mathbf{\tilde{c}}$ represents the perturbed concept vector.
\paragraph{Backdoor Injection.}
After identifying the optimal trigger $\tilde{\mathbf{c}}$ for the specified size, the attacker applies the poisoning function $T_e$ to the training data. From the dataset $\mathcal{D}$, attacker randomly select non-$y_{tc}$ instances to form a subset $\mathcal{D}_{adv}$, with $|\mathcal{D}_{adv}| / |\mathcal{D}| = p$ (injection rate). Applying $T_e:(\mathbf{x}_i, \mathbf{c}_i, y_i) \rightarrow (\mathbf{x}_i, \mathbf{c}_i \oplus \mathbf{\tilde{c}}, y_{tc})$ to each point in $\mathcal{D}_{adv}$ creates the poisoned subset $\tilde{\mathcal{D}}_{adv}$. We then retrain the CBMs with the modified training dataset $\mathcal{D}({T_\mathbf{e}}) = {\mathcal{D} + \tilde{\mathcal{D}}_{adv} - \mathcal{D}_{adv}}$.
\section{ConceptGuard}
\label{ConceptGuard}

\paragraph{Notation.} We use $\mathcal{D}$ to denote a dataset that consists of $n$ (input, concept, label)-pairs, i.e., $ \mathcal{D} = \left\{ \left( \mathbf{x}_1, \mathbf{c}_1, y_1 \right), \left( \mathbf{x}_2, \mathbf{c}_2,y_2 \right), \cdots,\left(\mathbf{x}_n, \mathbf{c}_n, y_n \right) \right\} $, where $\mathbf{c}_i$ is the concept vector of a input $\mathbf{x}_i$ and $y_i$ represents its label. We use $\mathcal{A}$ to denote a training algorithm that takes a dataset as input and produces a concept-to-label classifier. Given a testing concept vector, we use $f(\mathbf{c}_{test};\mathcal{D})$ to denote the predicted label of the concept-to-label classifier $f$ trained on the dataset $\mathcal{D}$ using the algorithm $\mathcal{A}$.

Now suppose $\mathbf{e}$ is a set of concepts used in backdoor trigger. Then we use $T_{\mathbf{e}}$ to denote the trigger injection by a concept-level backdoor attack. Given a concept vector $\mathbf{c}$, we use $\mathbf{c}' = T_{\mathbf{e}}(\mathbf{c})$ to denote a backdoored concept vector after the injection. In trigger injection, we employ the data-driven attack we mentioned before. Using the above notation we mentioned, we can use $\mathcal{D}(T_{\mathbf{e}},y_{tc},p)$ to denote the backdoored training dataset, which is created by injecting the backdoor trigger $T_{\mathbf{e}}$ to $p$ (injection rate) fraction of training instances in a clean dataset and relabeling them again as the target class $y_{tc}$. For simplicity, sometimes we write $\mathcal{D}(T_{\mathbf{e}})$ rather than $\mathcal{D}(T_{\mathbf{e}},y_{tc},p)$ instead when we focusing on the backdoor trigger, while less focus in target class and injection rate.

\begin{figure*}[t]
\centering
% \vspace{-10pt}
\includegraphics[width=1.00\textwidth]{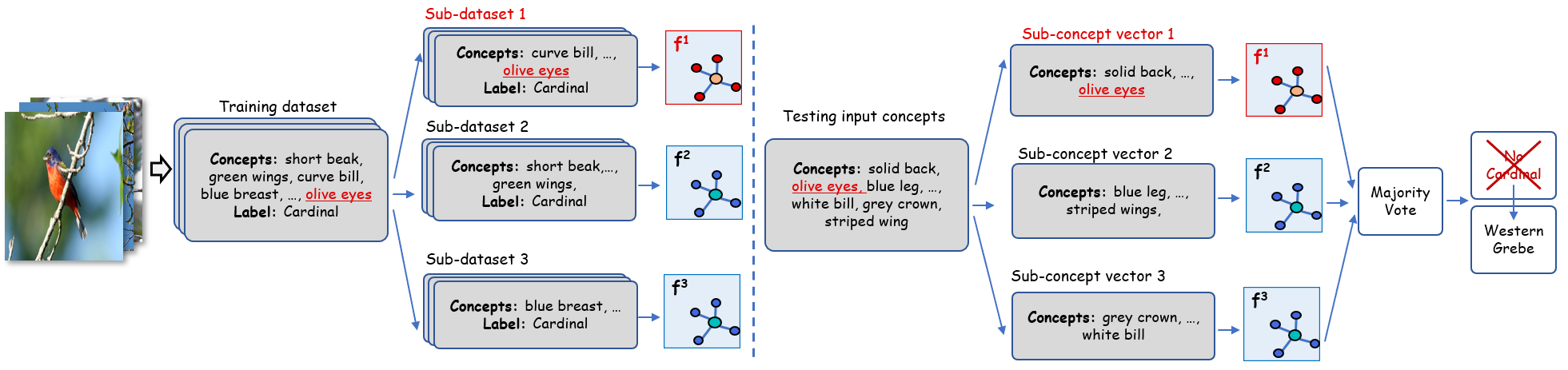}
\caption{Overview of ConceptGuard for concepts flow. Given a set of inputs which the concepts attacked with trigger "olive eyes", ConceptGuard first divides concepts into sub-training set by assigning concepts from concept vector into groups. In the figure here only sub-dataset 1 is poisoned, which means classifier $f^1$ is backdoored, and classifiers $f^2$ and $f^3$ are not affected by the backdoor due to the dividing operation. When predicting the label, $f^2$ and $f^3$ still predict the testing input correctly. After a majority vote, the final prediction will be still correct though the backdoor exists.}
% \vspace{-10pt}
\label{figure:defense}
\end{figure*}

\paragraph{Dividing concepts into groups.} 
Suppose that we have a concept vector $\mathbf{c} = \{c^1, c^2, \cdots, c^d\}$, where each $c^k (k = 1, 2, \cdots, d)$ is a specific concept and $d$ is the length of the concept vector, which is the number of the all concepts in the dataset. For each $c^k (k = 1, 2, \cdots, d)$, we firstly encode them into textual embeddings by some methods, such as TF-IDF, Word2Vec \cite{mikolov2013efficient}, Bert \cite{devlin2018bert}, then we can use some clustering algorithms to divide them into several groups, with the number of groups being $m$. This approach leads to clustering concepts that are semantically similar into the same category. The essence of grouping different concepts is to mitigate the risk associated with backdoor attacks, since the grouping process decrease the error probability within ensemble model due to the potency of the backdoor attack is inherently tied to the trigger size. We use $
\mathcal{G}^j\left( \mathbf{c} \right) $ to denote the concepts in a divided group, whose group index is $j$, where $j = 1, 2, \cdots, m$. The clustered concept groups $\mathcal{G}^j(\mathbf{c})$ such that $\bigcup_{j = 1}^m \mathcal{G}^j(\mathbf{c}) = \mathbf{c}$ and there's no overlap among them.

\paragraph{Constructing $m$ sub-datasets from a training dataset.}
Given an training dataset $\mathcal{D} = \left\{ \left( \mathbf{x}_1, \mathbf{c}_1,y_1 \right) , \left( \mathbf{x}_2, \mathbf{c}_2,y_2 \right) ,\cdots ,\left(\mathbf{x}_n, \mathbf{c}_n,y_n \right) \right\} $, where $n$ is the total number of training instances. Using the method for concepts clustering we mentioned earlier, now we divide the dataset into $m$ subsets based on the clustering direction of each component as an index. For each input training instance $(\mathbf{x}_i, \mathbf{c}_i, y_i) \in \mathcal{D}$, we can use clustering algorithm to divide $\mathbf{c}_i$ into $m$ groups: $\mathcal{G}^{1}(\mathbf{c}_i), \mathcal{G}^{2}(\mathbf{c}_i), \cdots, \mathcal{G}^{m}(\mathbf{c}_i$). Following the above grouping process and dataset, we can create $m$ (input, sub-concept, label) pairs: $\left( \mathbf{x}_i,\ \mathcal{G}^1\left( \mathbf{c}_i \right),\ y_i \right),\ \cdots,\ \left( \mathbf{x}_i,\ \mathcal{G}^m\left( \mathbf{c}_i \right),\ y_i \right) \ \ $. Finally, we use the group index $j$ to generate $m$ sub-datasets based on sub-concept. Specifically, we generate a sub-dataset $\mathcal{D}^j$ which consists of all (input, sub-concept, label) pairs whose group index is $j$, i.e., $\mathcal{D}^j = \{\left( \mathbf{x}_1,\ \mathcal{G}^j\left( \mathbf{c}_1 \right),\ y_1 \right), \cdots, \left( \mathbf{x}_n,\ \mathcal{G}^j\left( \mathbf{c}_n \right),\ y_n \right) \}$, where $j = 1, 2, \cdots, m$.

\paragraph{Building an Ensemble Concept-Based Classifier.}
Given sub-datasets, we can use an arbitrary training algorithm $\mathcal{A}$ to train a base concept-based classifier on each of the sub-datasets. We use $f^j$ to denote the base classifier trained on sub-dataset $\mathcal{D}^j$. Given a testing text $\mathbf{c}_{test}$, we also divide it into $m$ groups with the concept index, i.e., $\mathcal{G}^{1}(\mathbf{c}_{test}), \mathcal{G}^{2}(\mathbf{c}_{test}), \cdots, \mathcal{G}^{m}(\mathbf{c}_{test})$. Then we use the base classifier $f^j$ to predict the label of $\mathcal{G}^{j}(\mathbf{c}_{test})$. Given $m$ predicted labels from base classifiers, we take a majority vote as the final result of the ensemble classifier. Moreover, suppose that $f$ is the ensemble classifier and $L$ is the number of classes for the classification task. We define $N_l$ as the number of base classifiers which predicted the label $l$, i.e., $N_l = \sum\limits_{j=1}^m{\mathbb{I}\left( f^j\left( \mathcal{G}^j\left( \mathbf{c}_{test} \right) \right) =l \right)}$, where $\mathbb{I}$ is the indicator function, $l = 1, \cdots, L$. In total, our ensemble classifier is defined as:
\begin{equation}
\label{eq1}
    f(\mathbf{c}_{test};\mathcal{D}) = \underset{l = 1, 2, \cdots, L}{\mathrm{argmax}} \ N_l,
\end{equation}
and we take the smaller index label if there are any prediction ties.

Figure \ref{figure:defense} detects how the concepts flow in our ConceptGuard framework in training and testing specifically.

\section{Certified Robustness}
% thms and proof
In this section, we derive the certified size and certified accuracy of our ensemble classifier in concept-level. Suppose $\mathbf{c}_{test}$ is an clean testing input, we use $\mathbf{c}_{test}'$ to denote the backdoored concept vector created from $\mathbf{c}_{test}$ by $T_{\mathbf{e}}$. We will certify a classifier secure if $f(\mathbf{c}_{test}';\mathcal{D}(T_{\mathbf{e}}))$ is provably unaffected by the backdoor concept trigger $T_\mathbf{e}$, where the trigger size $|\mathbf{e}|$ is no larger than the threshold (\textit{certified size}). We use $\sigma(\mathbf{c}_{test})$ to denote the \textit{certified size} for concept vector $\mathbf{c}_{test}$. We formalize the following certain secure properties:
\begin{equation}
\label{eq:accu}
\begin{split}
    &f(\mathbf{c}_{test}';\mathcal{D}(T_{\mathbf{e}})) = f(\mathbf{c}_{test}; \mathcal{D}(\phi)), \forall \ |\mathbf{e}| \in \mathbb{R}, \ \text{s.t.} \ |\mathbf{e}| \leq \sigma(\mathbf{c}_{test}),
\end{split}
\end{equation}
where $\mathcal{D}(\phi)$ represents the original dataset without any trigger injecting(certified training dataset).
\paragraph{Deriving Certified Size for Concept Vector}
Suppose $N_l(\text{or} \ N_l ')$ is the number of the base classifiers that predict label $l$ for $\mathbf{c}_{test}(\text{or} \ \mathbf{c}_{test}')$ when the training dataset is $\mathcal{D}(\phi)(\text{or} \ \mathcal{D}(T_\mathbf{e}))$, where $l = 1,2,\cdots,L$. Now we first derive the bound for $\mathbf{c}_{test}'$. Note that each trigger concept in $\mathbf{e}$ belongs to a different single group as we use text distance measurements to determine the group index of each concept. It leads to $|\mathbf{e}|$ groups are corrupted by the backdoor trigger at most. So we have:
\begin{equation}
\label{eq3}
    N_l - |\mathbf{e}| \leq N_l ' \leq N_l + |\mathbf{e}|.
\end{equation}
Suppose that $y$ is our final prediction label from our ensemble classifier for $\mathbf{c}_{test}$ with $\mathcal{D}(\phi)$, i.e., $y = f(\mathbf{c}_{test};\mathcal{D}(\phi))$. From Equation \ref{eq1}, we can derive:
\begin{equation}
\label{eq4}
    N_y \geq \max_{l\ne y}(N_l + \mathbb{I}(y>l)),
\end{equation}
where $\mathbb{I}(y>l))$ because the classifier chooses a smaller index of the label when prediction ties. Based on Equation 4, the ensemble classifier using $\mathcal{D}(T_\mathbf{e})$ keep the prediction label $y$ unchanged if $N_y ' \geq \max_{l\ne y}(N_l' + \mathbb{I}(y>l))$. From Equation \ref{eq3}\&\ref{eq4}, we also have $N_y - |\mathbf{e}| \leq N_y '$, $\max_{l\ne y}(N_l' + \mathbb{I}(y>l)) \leq \max_{l\ne y}(N_l + \mathbb{I}(y>l) + |\mathbf{e}|)$. Then all we need to ensure will be $N_y - |\mathbf{e}| \geq \max_{l\ne y}(N_l + \mathbb{I}(y>l) + |\mathbf{e}|)$. In general, we keep prediction unchanged $f(\mathbf{c}_{test} '; \mathcal{D}(T_\mathbf{e}) = y$ if:
\begin{equation}
\label{eq5}
    |\mathbf{e}| \leq \frac{N_y - \max_{l\ne y}(N_l + \mathbb{I}(y>l))}{2}.
\end{equation}
We define \textit{certified size} $\sigma(\mathbf{c}_{test})$ as follows:
\begin{equation}
\label{eq6}
    \sigma(\mathbf{c}_{test}) = \frac{N_y - \max_{l\ne y}(N_l + \mathbb{I}(y>l))}{2}.
\end{equation}

The above derivation process is summarized as a theorem as follows:
\begin{theorem}[\textbf{Ensemble Classifier Certified Size}]
    \label{theo:1}
    Suppose $f$ is the ensemble concept classifier built by our defense framework. Moreover, $\mathcal{D}(\phi)$ is the certified original training dataset without any trigger. Given a testing concept vector $\mathbf{c}_{test}$, use $N_l$ to denote the number of the base classifiers trained on the sub-datasets created from $\mathcal{D}(\phi)$ which predict the label $l$, where $l=1,2,\cdots, L$. Assuming that $y$ is the final predicted label of the ensemble concept classifier built on $\mathcal{D}(\phi)$. Suppose $\mathbf{e}$ is a set of trigger concepts used in the backdoor attack. The predicted label is \textbf{PROVABLY UNAFFECTED} by the backdoor attack trigger when $|\mathbf{e}|$ is under certified size, i.e.
    \begin{equation}
    \label{eq7}
    \begin{split}
        &f(\mathbf{c}_{test}';\mathcal{D}(T_{\mathbf{e}})) = f(\mathbf{c}_{test}; \mathcal{D}(\phi)), \forall \ |\mathbf{e}| \in \mathbb{R}, \ \text{s.t.} \ |\mathbf{e}| \leq \sigma(\mathbf{c}_{test}),
    \end{split}
    \end{equation}
    where $\mathbf{e}_{test}'$ is the backdoored concept vector and $\sigma(\mathbf{c}_{test})$ is computed as follows:
    \begin{equation}
    \label{eq8}
        \sigma(\mathbf{c}_{test}) = \frac{N_y - \max_{l\ne y}(N_l + \mathbb{I}(y>l))}{2}. 
    \end{equation}
\end{theorem}
\begin{proof}
    See Appendix \ref{Appendix A}.
\end{proof}

\noindent \textbf{Summary:} Here we give a summary of the above defense theory:
\begin{itemize}
    \item[$\bigstar$] Our ConceptGuard is agnostic to the training classifier's algorithm $\mathcal{A}$ and the architecture of the model, allowing us to employ any training algorithm for each classifier while preserving the interpretability of CBMs.
    \item[$\bigstar$] Our ConceptGuard can provably defend against any concept-level backdoor attack when the trigger size $|\mathbf{e}|$ is not larger than a threshold.
    \item[$\bigstar$] Our ConceptGuard will exhibit enhanced performance, providing a larger certified size $\sigma(\mathbf{c}_{test})$ when the gap between $N_y$ and $\max_{l\ne y}(N_l + \mathbb{I}(y>l))$ is larger.
\end{itemize}

\paragraph{Independent Certified Accuracy}
Now we derive certified accuracy for a testing dataset by considering each testing text independently. Suppose $t$ is the maximum trigger size, i.e., $|\mathbf{e}| \leq t$. According to Theorem \ref{theo:1}, the predicted label of our ensemble classifier $f$ is provably unaffected by the backdoor trigger if the certified size $\sigma(\mathbf{c}_{test})$ is no smaller than $t$. Now we let $\mathcal{D}_{test}$ be a testing dataset. Given $t$ we define the \textit{certified accuracy} as a lower bound of the CBMs task accuracy that our ensemble classifier can achieve. Formally, we compute the independent certified accuracy as follows:
\begin{equation}
    Accu\left( \mathcal{D}_{test},t \right) = \\ 
    \\ \frac{\sum_{\left( \mathbf{c}_{test},y_{test} \right) \in \mathcal{D}_{test}}\mathbb{I}_{test}}{|\mathcal{D}_{test}|},
\end{equation}
\begin{equation}
    \mathbb{I}_{test} = {\mathbb{I}\left( f\left( \mathbf{c}_{test};\mathcal{D}\left( \phi \right) \right) =y_{test} \right)} \mathbb{I}\left( t\leq \sigma \left( \mathbf{c}_{test} \right) \right),
\end{equation}
where function $\mathbb{I}$ is the indicator function and $y_{test}$ is the ground truth of $\mathbf{c}_{test}$ (related to $\mathbf{x}_{test}$). We call above computing \textit{independent certification} because we consider each testing input $\mathbf{c}_{test}$ independently.

\paragraph{Improving Certified Accuracy Estimation}
Recall that for each test sample, we consider that the concepts in trigger $\mathbf{e}$ can arbitrarily corrupt a group of base classifiers in number of $|\mathbf{e}|$. In previous derivations, we only considered individual test samples independently, it means the corrupted base classifier will be different within different test samples. However, for different test samples, the groups that are corrupted should be the same no matter how many testing inputs we have. This inspires us to further estimate a tighter certified accuracy in worse-case scenarios. 

Specifically, when there are $m$ sub-datasets, the total number of possible combinations is given by $\binom{m}{|\mathbf{e}|}$, where $|\mathbf{e}|$ represents the size of the selected trigger. We assume that the $|\mathbf{e}|$ base classifiers chosen in each combination may be corrupted, and subsequently derive a potential accuracy for the testing dataset. Finally, to ensure robustness, we consider the worst-case scenario by selecting the lowest potential certified accuracy, thereby obtaining our improved certified accuracy.

We use $\mathcal{J}$ to denote the set of indices of $|\mathbf{e}|$ groups, which are potentially corrupted. We can derive the following lower and upper bounds for $N_c '$:
\begin{equation}
\label{eq15}
    N_l-\sum_{j\in \mathcal{J}}{\mathbb{I}\left( f\left( \mathcal{G}^j\left( \mathbf{c}_{test} \right) ;\mathcal{D}\left( \phi \right) \right) =l \right) \leq N_{l}^{'}},
\end{equation}
\begin{equation}
\label{eq16}
    N_{l}^{'} \leq N_l+\sum_{j\in \mathcal{J}}{\mathbb{I}\left( f\left( \mathcal{G}^j\left( \mathbf{c}_{test} \right) ;\mathcal{D}\left( \phi \right) \right) \ne l \right)},
\end{equation}

Intuitively, the lower bound (or upper bound) is obtained by having those base classifiers from potentially corrupted groups predict another class (or class $l$), if they originally predicted class $l$ (or another class). Suppose that $y$ is the predicted label of our ensemble classifier for $\mathbf{c}_{test}$ when we use $\mathcal{D}(\phi)$ to build. We conclude the property as the following theorem:
\begin{theorem}[\textbf{Improved joint Certified Accuracy}]
    \label{theo:2}
    Following the same notations in Theorem \ref{theo:1}, and the numbers of sub-datasets is $m$. The ensemble classifier f build upon $\mathcal{D(T_\mathbf{e})}$ still predicts the label y when :
    \begin{align*}
        &N_y-\sum_{j\in \mathcal{J}}{\mathbb{I}\left( f\left( \mathcal{G}^j\left( \mathbf{c}_{test} \right) ;\mathcal{D}\left( \phi \right) \right) =y \right)} \geq \max_{l\ne y}(N_l + \mathbb{I}(y>l)+ \sum_{j\in \mathcal{J}}\mathbb{I}\left( f\left( \mathcal{G}^j\left( \mathbf{c}_{test} \right);\mathcal{D}\left( \phi \right) \right) \neq l \right)).
    \end{align*}
    where $\mathcal{J}$ is denoted as the set of indices of $|\mathbf{e}|$ groups which are potentially corrupted. For $\mathcal{D}$ in each combination in $\binom{m}{|\mathbf{e}|}$, the improved accuracy could be computed as algorithm \ref{alg:joint certification}. The proof see in Appendix \ref{proof2}.
\end{theorem}

\begin{table*}[t]
  \centering
      \small 
    \aboverulesep = 0pt
    \belowrulesep = 0pt
            \begin{adjustbox}{max width=0.9\textwidth}
        \begin{tabular}{l|ccc|ccc}
    \hline
         & \multicolumn{3}{c|}{CUB} & \multicolumn{3}{c}{AwA} \\
    \cmidrule{2-7}
         & \multicolumn{1}{c}{Original ACC(\%)} & \multicolumn{1}{c}{ACC(\%)} & \multicolumn{1}{c|}{ASR(\%)} & \multicolumn{1}{c}{Original ACC (\%)} & \multicolumn{1}{c}{ACC(\%)} & \multicolumn{1}{c}{ASR(\%)} \\
         \hline
    CAT  &   \multirow{2}[0]{*}{81.65}   &   78.01   &  44.66    &   \multirow{2}[0]{*}{90.46}   &  87.86    & 48.24 \\
    CAT+ &      &   78.66   &   89.68   &      &   88.32   &  63.81\\
    ConceptGuard(CAT) & \multirow{2}[0]{*}{\textbf{83.03} $\uparrow$(1.38)}  &  78.75  &  \textbf{11.55} $\downarrow$ (74.10)  &  \multirow{2}[0]{*}{\textbf{91.30} $\uparrow$(0.84)}    &  91.20    & \textbf{13.68} $\downarrow$ (71.63)\\
    ConceptGuard(CAT+) &      &   78.56   & \textbf{17.16} $\downarrow$ (80.86) &      &   91.21   & \textbf{9.24} $\downarrow$ (85.52) \\
    \hline
    \end{tabular}%
    \end{adjustbox}
    \caption{The results for the evaluation of ConceptGuard. We fixed the injection rate $p$ of attack to 0.05 for both two datasets, and we fixed the trigger size $|\mathbf{e}|$ for CUB dataset to 20, for AwA dataset to 17. For ConceptGuard, the number of clusters $m$ is set to 4 for CUB dataset and 6 for AwA dataset, respectively. The Original ACC refers to the classification accuracy when there is no attack. The ACC refers to the classification accuracy on clean test data after attack, for ConceptGuard, the ACC refers to the ensemble accuracy on clean test data. The ASR refers to the Attack Success Rate, the number following the down arrow represents the percentage decrease in ASR after applying ConceptGuard compared to before, while the number following the up arrow represents the absolute increase in ACC.}
    % \vspace{-10pt}
  \label{table:1}%
\end{table*}%

\begin{table}[htbp]
  \centering
  % \vspace{-5pt}
    \begin{tabular}{l|ccc}
    \hline
         & \multicolumn{1}{c}{Original } & \multicolumn{1}{c}{CG(CAT)} & \multicolumn{1}{c}{CG(CAT+)} \\
    \hline
    Base model 1 &   \underline{77.61}   &  73.47    & 73.09 \\
    Base model 2 &   78.49  &  73.97    & 74.02 \\
    Base model 3 &   \textbf{81.34}   &   \textbf{77.05}   & \textbf{76.70} \\
    Base model 4 &   77.67   &  \underline{72.30}    & \underline{72.01} \\
    \hline
    Average &  78.78    &  74.20    &  73.96 \\ 
    Ensemble &  83.03 $\uparrow$   &  78.75 $\uparrow$ &  78.56$\uparrow$\\
    \hline
    \end{tabular}%
    \caption{The Accuracy (\%) for each sub-model on clean test data for CUB dataset, the Original denotes to the accuracy when there is no attack. The bold value refers to the best accuracy of sub-model and the underlined value refers to the worst accuracy of sub-model.}
  \label{tab:2}%
  % \vspace{-10pt}
\end{table}%

\begin{table}[htbp]
  \centering
   % \vspace{-5pt}
    \begin{tabular}{l|ccc}
    \hline
         & \multicolumn{1}{c}{Original } & \multicolumn{1}{c}{CG(CAT)} & \multicolumn{1}{c}{CG(CAT+)} \\
    \hline
    Base model 1 &  \underline{88.67}    & 87.34     & 87.50 \\
    Base model 2 &  89.52   &  \underline{86.13}    & \underline{86.44} \\
    Base model 3 &  89.79    &  86.82    & 86.49 \\
    Base model 4 &  \textbf{89.85}    &   86.73   & 86.54 \\
    Base model 5 &  88.88    &  86.79    & 87.02 \\
    Base model 6 &  88.94    &  \textbf{87.51}    &  \textbf{87.81}\\
    \hline
    Average &  89.28    &  86.89    &  86.97 \\ 
    Ensemble &  91.30 $\uparrow$    &   90.20 $\uparrow$   &  90.21 $\uparrow$\\
    \hline
    \end{tabular}%
    \caption{The Accuracy (\%) for each sub-model on clean test data for AwA dataset.}
  \label{tab:3}%
    \vspace{-15pt}
\end{table}%

\section{Experiments and Results}
\subsection{Datasets}
\noindent \textbf{CUB.} The Caltech-UCSD Birds-200-2011 (CUB) \cite{wah2011caltech} dataset is designed for bird classification and includes 11,788 images from 200 different species. It provides 312 binary attributes, offering high-level semantic information. Following the work in \cite{lai2024cat}, we filter out 116 attributes as the final concepts. To enhance the clustering process, we modify the format of these attributes at the textual level.

\noindent \textbf{AwA.} The Animals with Attributes (AwA) \cite{xian2018zero} dataset contains 37,322 images across 50 animal categories, each annotated with 85 binary attributes. To improve clustering effectiveness, we modify the concepts. Since each original concept is represented by a single word, we use GPT-4 \cite{achiam2023gpt} to generate full sentences to replace them. 

See Appendix \ref{Dataset_details} for more details and examples about the modifications of CUB and AwA.

\subsection{Settings}
We state brief experiments settings here and put the details in Appendix \ref{experiment settings}. For concepts clustering, we apply k-means to divide the concepts into $m$ groups, then we follow the method we introduced in Section \ref{ConceptGuard} to construct $m$ sub-datasets and train $m$ models individually. For each sub-model, the settings, including learning rate, optimizer, learning scheduler, and other parameters, are identical, and the model architectures are also the same, except for the input dimensions for the final prediction, see Appendix \ref{experiment settings} for more details.

\subsection{Experiment Results}
\subsubsection{ConceptGuard v.s. Direct Training (DT)}
To evaluate the effectiveness of ConceptGuard, we first compared it with direct training (DT). The results are presented in Table \ref{table:1}. Specifically, we conducted attacks with a fixed injection rate $p$ and trigger size $|\mathbf{e}|$ (various on different datasets). In general, our method achieves a significant decrease of ASR across all datasets, demonstrating its defense efficacy.
Notably, our method does not compromise the performance of the original tasks; in fact, it still outperforms the baseline model ($\uparrow 1\%-2\%$), which lacks a certified guarantee. Specifically, the attacked models without any triggers activated maintain similar accuracy to their original counterparts, indicating the imperceptibility of the CAT. This finding underscores the ability of ConceptGuard to preserve model utility under normal conditions while effectively defending against imperceptible attacks.
In terms of certified guarantees, our method achieved an ASR reduction of over 70\%, with the maximum reduction reaching 85.52\% on the AwA dataset when attacked by CAT+. These results confirm that ConceptGuard maintains the model's utility in the absence of attacks and provides strong defense against imperceptible attacks. The effectiveness of our approach can be attributed to its ability to disrupt the original patterns of backdoor triggers, thereby reducing the likelihood of the model memorizing the backdoor. Additionally, by introducing concept-level protection mechanisms, ConceptGuard ensures that the model remains effective and secure without compromising its normal performance.
The effectiveness of our approach can be attributed to its ability to disrupt the original pattern of the backdoor triggers, thereby reducing the likelihood of the model memorizing the backdoor.

\begin{table}[H]
  \centering
    \begin{tabular}{c|cc|cc}
    \hline
     & \multicolumn{2}{c|}{CUB} & \multicolumn{2}{c}{AwA} \\
    \hline
       $m$  & \multicolumn{1}{l}{CG(CAT)} & \multicolumn{1}{l|}{CG(CAT+)} & \multicolumn{1}{l}{CG(CAT)} & \multicolumn{1}{l}{CG(CAT+)} \\
    1    & 44.66 & 89.68 & 48.24 & 63.81 \\
    3    & 30.78 & 42.75 & 28.84 & 57.56 \\
    4    & \underline{11.55} & \underline{17.16} & 48.77 & \underline{5.36} \\
    5    & 25.95 & \textbf{16.64} & 10.54 & 67.54 \\
    6    & 23.84 & 20.12 & 13.68 & 9.24 \\
    7    & 15.41 & 24.50 & 17.71 & 5.66 \\
    8    & 17.70 & 30.92 & 9.90 & 5.87 \\
    9    & 15.23 & 25.35 & \underline{7.49} & 9.96 \\
    10   & \textbf{10.22} & 19.33 & \textbf{3.73} & \textbf{5.15} \\
    \hline
    \end{tabular}%
    \caption{Attack Success Rate (ASR, \%) under varying numbers of clusters $m$. CG denotes ConceptGuard. Bold \textbf{values} highlight the best performance, while underlined \underline{values} indicate competitive performance. $m = 1$ refers to the ASR when ConceptGuard is not applied.}
  \label{tab: tab4}%
\end{table}%

\begin{figure}[htp]
\centering
% % \vspace{-10pt}
\includegraphics[width=0.6\textwidth]{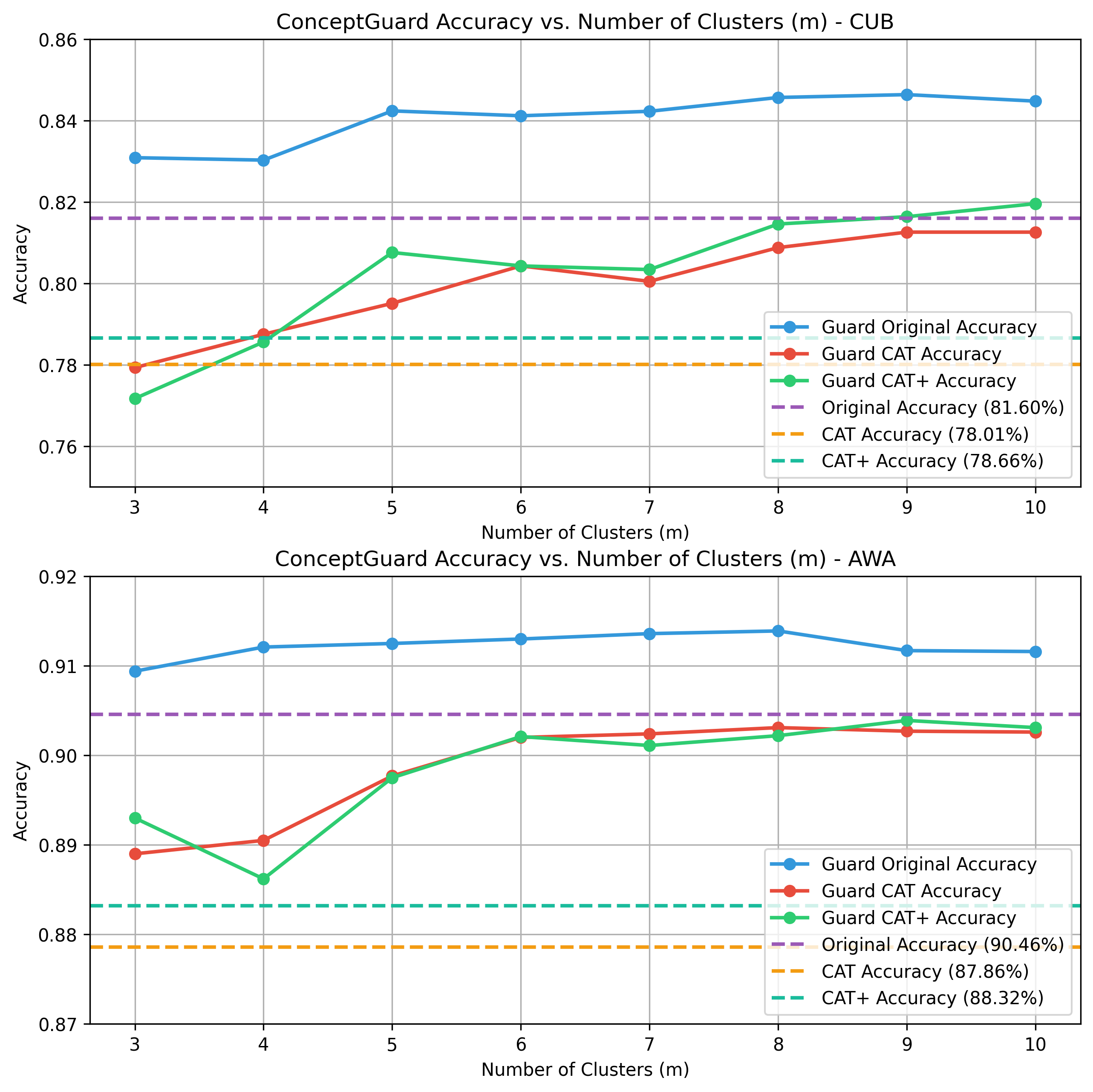}
% \vspace{-10pt}
\caption{The ConceptGuard Accuracy versus the number of Clusters $m$, the Guard Original Accuracy (blue lines) denotes to the accuracy when there is no attack, and Guard CAT\textbackslash CAT+ Accuracy (red lines\textbackslash
 green lines) denotes to the accuracy when CAT \textbackslash CAT+ is applied, }
 \vspace{-10pt}
\label{figure:conceptguard_accuracy_plot}
\end{figure}

\subsubsection{Individual model vs. ensemble model}
We investigated the individual model's accuracy and the ensembled accuracy, with the results presented in Table \ref{tab:2} and Table \ref{tab:3}. Overall, our ensemble model shows a significant improvement in accuracy compared to the individual accuracy of each base classifier in all scenarios, even outperforming the best-performing base classifier (base classifier 3). Additionally, there is a notable increase in accuracy compared to the average accuracy of the base classifiers. This improvement is attributed to our ConceptGuard framework, which effectively filters out the misclassifications of the few base classifiers during testing, thereby providing the ensemble model with a higher accuracy. The source of this accuracy improvement aligns with the original motivation of ConceptGuard: it mitigates the errors of the base classifiers, leading to a higher ensemble accuracy, rather than simply relying on a straightforward aggregation of the classifiers. By leveraging the diversity and robustness of the ensemble, ConceptGuard ensures that the final predictions are more accurate and reliable, demonstrating its effectiveness in enhancing the performance of ensemble models.

\subsubsection{The impact of number of clusters}
We further evaluate ConceptGuard with different settings of the number of clusters $m$, the experimental results are shown in Table \ref{tab: tab4} and Figure \ref{figure:conceptguard_accuracy_plot}. We observe that as $m$ increases, the Attack Success Rate (ASR) generally decreases, indicating that dividing the dataset into more groups helps mitigate the backdoor effect more effectively. Meanwhile, the Accuracy generally increases as the increase of $m$, and even exceed the performance before attack, for example, when $m$ is set to 10, the Accuracy for ConceptGuard(CAT+) exceeds the original accuracy. However, choosing an excessively large $m$ is not practical, as the computational cost increases approximately linearly with $m$.

\section{Limitation}
Despite the significant contributions of ConceptGuard in enhancing the security and trustworthiness of CBMs against concept-level backdoor attacks, there remain several limitations to consider. Firstly, while ConceptGuard demonstrates effectiveness in defending against backdoor attacks within a certain trigger size threshold, the exact boundary of this threshold may vary across different datasets and application domains, necessitating further research to generalize its applicability. Secondly, the computational cost associated with the multi-stage approach, including concept clustering and the training of multiple sub-models, poses a challenge for real-time or resource-constrained environments. Although increasing the number of clusters can improve both the attack success rate and overall accuracy, there is a trade-off with computational efficiency, highlighting the need for optimized algorithms that balance performance and resource utilization. Additionally, ConceptGuard's reliance on concept clustering based on text distance measurements assumes that relevant concepts are semantically close, which may not always hold true, especially in complex or highly specialized domains where the semantic relationships between concepts are less clear.

\section{Conclusion}

ConceptGuard represents a significant advancement in the field of secure and explainable artificial intelligence, specifically addressing the critical issue of concept-level backdoor attacks in CBMs. By introducing a novel defense framework that leverages concept clustering and a voting mechanism among classifiers trained on different concept subgroups, ConceptGuard not only mitigates the risks posed by such attacks but also maintains the high performance and interpretability of CBMs. Theoretical analyses and empirical evaluations have demonstrated the effectiveness of ConceptGuard in enhancing the robustness of CBMs, making them more reliable and trustworthy for deployment in high-stakes applications such as medical diagnostics and financial services. While the current work has laid a solid foundation for defending against concept-level backdoors, future research should aim to address the identified limitations, such as optimizing computational efficiency, exploring alternative clustering methods, and expanding the scope of protection to encompass a broader range of potential threats.

%Bibliography  
\bibliography{references}

% \clearpage
\appendix
\section{Proof of Theorem \ref{theo:1}}
\label{Appendix A}

\begin{proof}
Our ConceptGuard clusters the concept components into groups within the concept vector first. After grouping, each concept appears exclusively in one group, implying that a backdoor trigger can corrupt $|\mathbf{e}|$ group at most. When the trigger size is less than $t$, i.e., $|e| \leq t$, at most $t$ groups are corrupted. Therefore, we can derive the dual bounds.
\begin{equation}
\label{eq9}
    N_l - |\mathbf{e}| \leq N_l ' \leq N_l + |\mathbf{e}|, l = 1, 2, \cdots, L,
\end{equation}
where $N_l '$ is the number of the base classifiers that predict the label $l$ built upon the dataset $\mathcal{D}(T_\mathbf{e})$. We mentioned that $y$ is the final predicted label of ensemble classifier for $\mathbf{c}_{test}$ with no attack, i.e., $y = f(\mathbf{c}_{test};\mathcal{D}(\phi))$. From Equation \ref{eq4}, the ensemble classifier built upon $\mathcal{D}(T_{\mathbf{e}})$ keep the prediction $y$ unchanged if the condition is satisfied: $N_y ' \geq \max_{l\ne y}(N_l' + \mathbb{I}(y>l))$. From Equation \ref{eq3}\&\ref{eq4}, we conclude that
$N_y - |\mathbf{e}| \leq N_y '$, $\max_{l\ne y}(N_l' + \mathbb{I}(y>l)) \leq \max_{l\ne y}(N_l + \mathbb{I}(y>l) + |\mathbf{e}|)$. Therefore, our primary objective is to ensure that: $N_y - |\mathbf{e}| \geq \max_{l\ne y}(N_l + \mathbb{I}(y>l) + |\mathbf{e}|)$. It makes the ensemble classifier predict the label $y$ still. Equivalently, $f(\mathbf{c}_{test}';\mathcal{D}(T_{\mathbf{e}})) = y$ if:
\begin{equation}
\label{eq10}
    |\mathbf{e}| \leq \frac{N_y - \max_{l\ne y}(N_l + \mathbb{I}(y>l))}{2}.
\end{equation}
\end{proof}

\section{Attack Formulation and Details}
\label{Appendix B}
In our attack formulation, we first recall our motivation and give the following definition:
\begin{equation}
\begin{split}
    \max_{\mathcal{D}^j \in \mathcal{D}} \Sigma_{\mathcal{D}^j}(f(\mathbf{c}_j)-f(\mathbf{c}_j \oplus \mathbf{\tilde{c}}))  \\
    \mathrm{s.t.} \quad f(\mathbf{c}_j) = f(\mathbf{c}_j \oplus \mathbf{\tilde{c}})) = y_{tc},
\end{split}
\end{equation}
where $\mathcal{D}^j$ represents each data, $\mathcal{D}$ represents the dataset, $y$ represents the clean-label which we chose to attack, and $\mathbf{c} \oplus \mathbf{\tilde{c}}$ represents the Data-Driven Attack Pattern we defined, it means we may change the concepts values in the concepts which we filtered out while we keep the values unchanged in other concepts.

The objective function during an attack is to maximize the discrepancy in predictions. Nevertheless, if the trigger is absent from the concept vector, the predicted label will remain unchanged. Crucially, the objective function adheres to two constraints: the first ensures that the model's predictions for the original dataset remain unchanged, while the second mandates that the perturbation remains imperceptible.

In concept-level backdoor attacks, the core mechanism consists of two steps: concepts filter for the attack and inject poisonous data into the training dataset to embed the backdoor trigger. Below, we will discuss how these two steps influence concept-level backdoor attacks.
\paragraph{Concept Filter.} 
Given a concept vector $\mathbf{c} \in \mathbb{R}^{L}$, where each element $c^k$ encapsulates a distinct concept, we endeavor to filter out the most irrelevant concepts to generate perturbations in the context of the attack. Let $\mathbf{e}$ represent a set of concepts, termed \textit{trigger concepts}, employed in the formulation of the backdoor trigger, such that $\mathbf{e} = \{ c^{k_1}, c^{k_2}, \ldots, c^{k_{|\mathbf{e}|}} \}$. Here, $|\mathbf{e}|$ denotes the cardinality of the concept set $\mathbf{e}$ and is defined as the \textit{trigger size}. During the concept filtering process, we systematically identify and eliminate the $|\mathbf{e}|$ concepts that exhibit the least relevance to the prediction task. The assessment of concept irrelevance is conducted through the utilization of the classifier $f$. Ultimately, we extract $|\mathbf{e}|$ concepts to facilitate the attack in subsequent stages. In this filtering process, the potency of the backdoor attack is inherently tied to the trigger size $|\mathbf{e}|$. We denote the resultant filtered concepts as $\tilde{\mathbf{c}}$.

\paragraph{Data-Driven Attack Pattern.}
In CBM tasks, most datasets have sparser concept levels in the concept bottleneck layer. It means in a concept vector $\textbf{c}$, most concept levels $c^k$ are positive (negative) rather than negative (positive). When $c^k = 0$, $c^k$ is negative, and when $c^k = 1$, $c^k$ is positive. Different datasets have different levels of sparsity. While we attacking the positive datasets, we set the filtered concepts $\mathbf{\tilde{c}}$ into $\mathbf{0}$, i.e., $\mathbf{\tilde{c}}:= {\left\{ 0,0,\cdots, 0 \right\} }, |\tilde{\mathbf{c}}|=|\mathbf{e}|$. There will be an opposite situation in negative datasets, we set the filtered concepts $\mathbf{\tilde{c}}$ into $\mathbf{1}$, i.e., $\mathbf{\tilde{c}}:= {\left\{ 1,1,\cdots, 1 \right\}, |\tilde{\mathbf{c}}|=|\mathbf{e}| }$.

By setting the filtered concepts $\mathbf{\tilde{c}}$ accordingly, the attack aims to introduce perturbations that are subtle yet impactful, disrupting the model's predictions without being easily detected. 
\paragraph{CAT+ \cite{lai2024cat}.} Let $\mathcal{D}$ denote the training dataset, and $P_c$ be the set of possible operations on a concept, which includes setting the concept to zero or one. 
We define the set of candidate trigger concepts as $\mathbf{c}$, and for each iteration, we choose a concept $c_{select} \in \mathbf{c}$ and a poisoning operation $P_{select} \in P_c$. The objective is to maximize the deviation in the label distribution after applying the trigger. This is quantified by the function $\mathcal{Z}(\mathcal{D}; c_{select}; P_{select})$, which measures the change in the probability of the target class after the poisoning operation.

The function $\mathcal{Z}(\cdot)$ is defined as follows:

(i) Let $n$ be the total number of training samples, and $n_{target}$ be the number of samples from the target class. The initial probability of the target class is $p_0 = n_{target} / n$.

(ii) Given a modified dataset $c_a = \mathcal{D}; c_{select}; P_{select}$, we calculate the conditional probability of the target class given $c_a$ as $p^{(target|c_a)} = \mathbb{H}({target}(c_a)) / \mathbb{H}(c_a)$, where $\mathbb{H}$ is a function that computes the overall distribution of labels in the dataset.

 (iii) The Z-score for $c_a$ is defined as:
 \begin{align*}
    \mathcal{Z}(c_a) &= \mathcal{Z}(c_{select}, P_{select}) \\
    &= \left[ p^{(target|c_a)} - p_0 \right] / \left[ \frac{p_0(1-p_0)}{p^{(target|c_a)}} \right]
 \end{align*}

  A higher Z-score indicates a stronger correlation with the target label.

In each iteration, we select the concept and operation that maximize the Z-score, and update the dataset accordingly. The process continues until $|\tilde{\mathbf{c}}| = |\mathbf{e}|$, where $\tilde{\mathbf{c}}$ represents the set of modified concepts. Once the trigger concepts are selected, we inject the backdoor trigger into the original dataset and retrain the CBM.

\section{Pseudo Algorithm}
\begin{algorithm}[h]
\caption{ConceptGuard Defense Algorithm}
\begin{algorithmic}[1]
\State \textbf{Input:} The concept vector $\mathbf{c}_{test}$, the training dataset $\mathcal{D}$, the backdoor trigger size $|\mathbf{e}|$, the number of sub-datasets $m$
\State \textbf{Output:} Improved certified accuracy for the ensemble classifier

\State Compute $N_l$ for each class label $l$ from the training dataset $\mathcal{D}(\phi)$
\State Compute the predicted label $y$ from the ensemble classifier on the clean data $\mathcal{D}(\phi)$
\State Calculate the maximum number of classifiers for each class $l$:
\[
    \max_{l\ne y}(N_l + \mathbb{I}(y>l))
\]

\For{each subset $\mathcal{G}^j(\mathbf{c}_{test})$ for $j = 1, \dots, m$}
    \State Compute the number of base classifiers for each class $l$ for the subset $\mathcal{G}^j(\mathbf{c}_{test})$
    \If{$f(\mathcal{G}^j(\mathbf{c}_{test}); \mathcal{D}(\phi)) = l$}
        \State Increase the counter for the predicted class $l$
    \EndIf
\EndFor

\For{each possible corrupted group index set $\mathcal{J} \subseteq \{1, 2, \dots, m\}$ with $|\mathcal{J}| = |\mathbf{e}|$}
    \State Compute the updated prediction $N_l'$ after corrupting the groups in $\mathcal{J}$:
    \State \quad $N_l' \leq N_l + \sum_{j \in \mathcal{J}} \mathbb{I}(f(\mathcal{G}^j(\mathbf{c}_{test}); \mathcal{D}(\phi)) \neq l)$
    \State \quad $N_l' \geq N_l - \sum_{j \in \mathcal{J}} \mathbb{I}(f(\mathcal{G}^j(\mathbf{c}_{test}); \mathcal{D}(\phi)) = l)$
    \State Check if the inequality for maintaining label $y$ is satisfied:
    \begin{align*}
            &N_y - \sum_{j \in \mathcal{J}} \mathbb{I}(f(\mathcal{G}^j(\mathbf{c}_{test}); \mathcal{D}(\phi)) = y) \geq \\ \nonumber
            &\max_{l \neq y}(N_l + \mathbb{I}(y>l) + \sum_{j \in \mathcal{J}} \mathbb{I}(f(\mathcal{G}^j(\mathbf{c}_{test}); \mathcal{D}(\phi)) \neq l))
    \end{align*}
    \If{the condition holds}
        \State \textbf{Accept} this subset as contributing to the certified accuracy
    \Else
        \State \textbf{Reject} this subset
    \EndIf
\EndFor

\State Compute the final certified accuracy by taking the majority vote over all subsets
\State \textbf{Output:} Certified accuracy
\label{al_gard}
% \vspace{-10pt}
\end{algorithmic}
\end{algorithm}

See in Algorithm \ref{alg:joint certification}.
\begin{algorithm}
\caption{Joint Certification}
\label{alg:joint certification}
\begin{algorithmic}[1]
\Require $m$ base classifiers $f^j$ $(j = 1, 2, \dots, m)$, a clustering function $\mathcal{F}$, a clean test dataset $\mathcal{D}_{\text{test}}$, maximum trigger size $t$.
\State $Accu \leftarrow 1$
\For{$\mathcal{J}$ in $\text{Combination}(m, t)$}
    \State $Iaccu \leftarrow 0$
    \For{($\mathbf{x}_{\text{test}}, \mathbf{c}_{\text{test}}, y_{\text{test}}) \in \mathcal{D}_{\text{test}}$}
        \State $\mathcal{G}^j(\mathbf{c}_{\text{test}}) \leftarrow \text{ConceptClustering}(\mathbf{c}_{\text{test}}, m, \mathcal{F})$, $j = 1, 2, \dots, m$
        \State $N_l \leftarrow \sum_{j=1}^m \mathbb{I}(f^j(\mathcal{G}^j(\mathbf{c}_{\text{test}}); \mathcal{D}(\phi)) = l), \quad l = 1, 2, \dots, C$
        \State $y \leftarrow \arg\max_{l=1,2,\dots,L} N_l$
        \State $U \leftarrow N_y - \sum_{j \in \mathcal{J}} \mathbb{I}(f^j(\mathcal{G}^j(\mathbf{c}_{\text{test}}); \mathcal{D}(\phi)) = y)$
        \State $L \leftarrow \max_{l\ne y}(N_l + \mathbb{I}(y>l)+ \sum_{j\in \mathcal{J}}\mathbb{I}\left( f\left( \mathcal{G}^j\left( \mathbf{c}_{test} \right);\mathcal{D}\left( \phi \right) \right) \neq l \right))$
        \State $Iaccu \leftarrow Iaccu + \mathbb{I}(U \geq L) \mathbb{I}(y_{\text{test}} = y)$
    \EndFor
    \State $Accu \leftarrow \min(Accu, Iaccu) / |\mathcal{D}_\text{test}|$
\EndFor
\State \textbf{return} $Accu$
\end{algorithmic}
\end{algorithm}

\section{Proof of Theorem \ref{theo:2}}
\label{proof2}
\begin{proof}
    Following the same notation, we use $\mathcal{J}$ to denote the set of indices of $|\mathbf{e}|$ groups which are potentially corrupted. When the groups with their indices in $\mathcal{J}$ are corrupted, the lower and upper bounds for $N_l'$ are derived as below:
    \begin{equation*}
        N_l-\sum_{j\in \mathcal{J}}{\mathbb{I}\left( f\left( \mathcal{G}^j\left( \mathbf{c}_{test} \right) ;\mathcal{D}\left( \phi \right) \right) =l \right) \leq N_{l}^{'}},
    \end{equation*}
    \begin{equation*}
        N_{l}^{'} \leq N_l+\sum_{j\in \mathcal{J}}{\mathbb{I}\left( f\left( \mathcal{G}^j\left( \mathbf{c}_{test} \right) ;\mathcal{D}\left( \phi \right) \right) \ne l \right)}.
    \end{equation*}
    Based on equation \ref{eq10} and Appendix \ref{Appendix A}, the ensemble classifier $f$ built upon $\mathcal{D}(\phi)$ still predicts $y$ for $\mathbf{c}_{test}$ if we have $N_y-\sum_{j\in \mathcal{J}}\mathbb{I}\left( f\left( \mathcal{G}^j\left( \mathbf{c}_{test} \right);\mathcal{D}\left( \phi \right) \right) =l \right) \geq \max_{l\neq y}(N_y+\sum_{j\in \mathcal{J}}\mathbb{I}\left( f\left( \mathcal{G}^j\left( \mathbf{c}_{test} \right);\mathcal{D}\left( \phi \right) \right) \neq l \right))$. Based on Equations \ref{eq15} and \ref{eq16}, the ensemble classifier $f$ built upon $\mathcal{D}(T_{\mathbf{e}})$ still predicts $y$ for $\mathbf{c}_{\text{test}}$ if we have:
    \begin{align*}
        &N_y-\sum_{j\in \mathcal{J}}{\mathbb{I}\left( f\left( \mathcal{G}^j\left( \mathbf{c}_{test} \right) ;\mathcal{D}\left( \phi \right) \right) =y \right)} \geq \\ 
        &\max_{l\ne y}(N_l + \mathbb{I}(y>l)+ \sum_{j\in \mathcal{J}}\mathbb{I}\left( f\left( \mathcal{G}^j\left( \mathbf{c}_{test} \right);\mathcal{D}\left( \phi \right) \right) \neq l \right)).
    \end{align*}
\end{proof}

\section{Dataset Details}
\label{Dataset_details}
%%%%% 放CUB和AWA的细节处理和样例展示
Here we give some examples of the modified concepts for the datasets, see Table \ref{tab:3}. For CUB dataset, we just change the format of the concepts. For AwA dataset, we use GPT-4 to generate one full sentence based on the single word concept through the following prompt, "Here are the concepts for an animal classification task, please transfer each concept into one complete sentence."

\begin{table}[htbp]
  \centering
  \aboverulesep = 0pt
  \belowrulesep = 0pt
    \begin{tabularx}{\linewidth}{c|X|X}
    \toprule
        Dataset &  \multicolumn{1}{c|}{Original concept} & \multicolumn{1}{c}{Rewrite concept} \\
    \midrule
    \multirow{2}{*}[-3.0ex]{CUB} & has\_bill\_\newline shape::dagger  & Bill shape is dagger \\
\cmidrule{2-3}         &  has\_eye\_\newline color::black    & Eye color is black \\
    \midrule
    \multirow{2}{*}{AWA} &  meat  & The animal consumes meat as part of its diet \\
\cmidrule{2-3}         &  Forest    & The animal inhabits forests  \\
    \bottomrule
    \end{tabularx}%
    \caption{Examples of the rewrite concepts for both datasets}
  \label{tab:tab5}%
\end{table}%

\section{Experiment Settings}
\label{experiment settings}
We conducted all of our experiments on a NVIDIA A800 GPU. The hyper-parameters for each dataset and for each sub-model remained consistent, regardless of whether an attack was present.

In this work, we set the training model in CBMs as joint bottleneck training, which minimizes the weighted loss function:
\begin{equation}
\label{eq19}
\begin{split}
    \hat{f},\hat{g} &= \arg \min_{f,g}\Sigma_{i}[L_y(f(g(x^{(i)}));y^{(i)}) \\
    &+\Sigma_{j}\lambda L_{c^j}(g(x^{(i)});c^{(i)})],
\end{split}
\end{equation}
where $\lambda > 0$, and loss function $L_y: \mathbb{R} \times \mathbb{R} \rightarrow \mathbb{R}_{+}$ measure the discrepancy between predicted and true targets, loss function $L_{c^j}: \mathbb{R} \times \mathbb{R} \rightarrow \mathbb{R}_{+}$  measures the discrepancy between the predicted and true $j$-th concept.

For model architecture, We use ResNet-50 \cite{he2016deep} as the Encoder to map the image to concept space, and then an MLP with one hidden layer, whose hidden size 512 is followed to make the final prediction. For one sub-model, the input dimension for the MLP will be the number of concepts in the corresponding group.

During training, we use a batch size of 64 and a learning rate of 1e-4. The Adam optimizer is applied with a weight decay of 5e-5, alongside an exponential learning scheduler with $\gamma$ = 0.95. The concept loss weight $\lambda$ in Equation \ref{eq19} is set to 0.5. For image augmentations, we follow the approach of \cite{lai2024cat}. Each training image is augmented using random color jittering, random horizontal flips, and random cropping to a resolution of 256. During inference, the original image is center-cropped and resized to 256. For AwA dataset, We use a batch size of 128, while all other hyper-parameters and image augmentations remain consistent with those used for the CUB dataset.

\section{More Experiments about other Setting}
\label{sec:more1}

% \usepackage{multirow}
\begin{table*}[]
  \resizebox{1 \columnwidth}{!}{%
\begin{tabular}{ccccccccccccc}
\hline
\multirow{3}{*}{} & \multicolumn{4}{c}{2\%}                                  & \multicolumn{4}{c}{}                                     & \multicolumn{4}{c}{10\%}                                 \\ \cline{2-13} 
                  & \multicolumn{2}{c}{CAT}  & \multicolumn{2}{c}{CAT (CG)}  & \multicolumn{2}{c}{CAT}  & \multicolumn{2}{c}{CAT (CG)}  & \multicolumn{2}{c}{CAT}  & \multicolumn{2}{c}{CAT (CG)}  \\ \cline{2-13} 
                  & ACC(\%)     & ASR(\%)    & ACC(\%)        & ASR(\%)      & ACC(\%)     & ASR(\%)    & ACC(\%)        & ASR(\%)      & ACC(\%)     & ASR(\%)    & ACC(\%)        & ASR(\%)      \\ \hline
12                & 80.72       & 13.97      & 81.77          & 12.3         & 78.7        & 24.05      & 80.36          & 11.45        & 74.66       & 38.08      & 75.22          & 27.55        \\
15                & 80.22       & 11.94      & 82.05          & 10.01        & 78.08       & 22.97      & 80.01          & 10.91        & 74.02       & 38.72      & 74.53          & 30.53        \\
17                & 80.31       & 25.07      & 82.15          & 3.94         & 78.86       & 46.69      & 79.75          & 16.66        & 73.27       & 61.28      & 74.77          & 20.32        \\
20                & 80.2        & 30.33      & 81.93          & 15.29        & 78.01       & 44.66      & 78.75          & 11.55        & 73.85       & 60.48      & 76.15          & 38.5         \\
23                & 80.31       & 20.42      & 82.21          & 23.28        & 78.06       & 32.48      & 80.26          & 25.56        & 72.63       & 47.02      & 75.72          & 48.2         \\ \hline
\multirow{2}{*}{} & \multicolumn{2}{c}{CAT+} & \multicolumn{2}{c}{CAT+ (CG)} & \multicolumn{2}{c}{CAT+} & \multicolumn{2}{c}{CAT+ (CG)} & \multicolumn{2}{c}{CAT+} & \multicolumn{2}{c}{CAT+ (CG)} \\ \cline{2-13} 
                  & ACC(\%)     & ASR(\%)    & ACC(\%)        & ASR(\%)      & ACC(\%)     & ASR(\%)    & ACC(\%)        & ASR(\%)      & ACC(\%)     & ASR(\%)    & ACC(\%)        & ASR(\%)      \\ \hline
12                & 80.46       & 38.88      & 82.34          & 21.69        & 79.05       & 57.6       & 80.07          & 34.49        & 75.11       & 68.49      & 76.35          & 44.69        \\
15                & 80.26       & 31.97      & 82.27          & 14.68        & 79.34       & 41.64      & 79.79          & 38.46        & 74.78       & 47.29      & 74.87          & 43.84        \\
17                & 79.84       & 49.22      & 82.29          & 31.94        & 78.48       & 58.31      & 80.45          & 38.88        & 73.85       & 71.2       & 75.77          & 31.21        \\
20                & 81.27       & 72.36      & 82.36          & 11.78        & 78.86       & 89.68      & 78.56          & 17.16        & 74.34       & 92.4       & 75.89          & 34.21        \\
23                & 79.58       & 87.4       & 81.88          & 22.38        & 77.65       & 91.71      & 79.1           & 42.37        & 73.78       & 86.9       & 76.87          & 40.49        \\ \hline
\end{tabular}
}
  \centering
  \caption{Performance Comparison of ConceptGuard under Different Injection Rates and Trigger Sizes. This table presents ACC and ASR of ConceptGuard under different injection rates (2\% and 10\%) and trigger sizes (12, 15, 17, 20, 23). The results are shown for both the unprotected models (CAT/CAT+) and the models protected by ConceptGuard (CAT (CG)/CAT+ (CG)).}
\label{tab:ij_ts}
\end{table*}

% \usepackage{multirow}
\begin{table*}[htp]
  \centering
\begin{tabular}{ccccccccc}
\hline
\multirow{2}{*}{Target Class} & \multicolumn{2}{c}{CAT} & \multicolumn{2}{c}{CAT (CG)} & \multicolumn{2}{c}{CAT+} & \multicolumn{2}{c}{CAT+ (CG)} \\ \cline{2-9} 
                              & ACC(\%)    & ASR(\%)    & ACC(\%)       & ASR(\%)      & ACC(\%)     & ASR(\%)    & ACC(\%)       & ASR(\%)       \\ \hline
8                             & 75.06      & 74.24      & 79.72         & 21.93        & 75.56       & 52.63      & 79.96         & 11.19         \\
16                            & 75.16      & 40.91      & 80.36         & 8.74         & 75.72       & 68.81      & 80.13         & 30.52         \\
24                            & 74.37      & 35.48      & 80.03         & 17.24        & 74.91       & 54.48      & 81            & 13.38         \\
32                            & 74.7       & 37.68      & 80.62         & 25.31        & 75.58       & 17.87      & 79.53         & 9.52          \\
40                            & 74.46      & 42.35      & 79.53         & 7.93         & 74.96       & 23.77      & 80.32         & 12.49         \\
48                            & 75.09      & 49.77      & 80.74         & 4.49         & 75.73       & 95.11      & 80.91         & 22.14         \\
56                            & 75.22      & 70.99      & 81.07         & 15.93        & 75.23       & 57.36      & 80.24         & 15.25         \\
64                            & 74.85      & 43.63      & 80.41         & 10.27        & 74.58       & 84.27      & 80.46         & 19.29         \\
72                            & 74.99      & 47.99      & 80.67         & 10.9         & 75.34       & 59.06      & 80.62         & 10.11         \\
80                            & 75.03      & 62.51      & 80.89         & 10.83        & 75.61       & 75.83      & 80.45         & 11.85         \\
88                            & 74.75      & 51.13      & 79.94         & 21.93        & 74.66       & 72.71      & 80.1          & 25.02         \\
96                            & 74.82      & 17.73      & 80.67         & 8.99         & 75.2        & 62.16      & 80.91         & 14.78         \\
104                           & 74.84      & 53.02      & 80.65         & 12.8         & 74.92       & 40.42      & 80            & 16.03         \\ \hline
\end{tabular}
\caption{Performance of ConceptGuard on Different Target Classes. This table presents ACC and ASR of ConceptGuard for different target classes under an injection rate of 5\% and a trigger size of 20, using the CUB dataset. The results are shown for both the unprotected models (CAT/CAT+) and the models protected by ConceptGuard (CAT (CG)/CAT+ (CG)).}
  \centering
\label{app:tc}
\end{table*}

\begin{table}[htp]
  \centering
    \begin{tabular}{c|cc}
    \hline
     & \multicolumn{2}{c}{CUB} \\
    \hline
       $m$  & \multicolumn{1}{l}{CG(CAT)} & \multicolumn{1}{l}{CG(CAT+)} \\
    1    & 60.48 & 92.40 \\
    3    & 41.67 & 57.46 \\
    4    & 38.50 & 34.21 \\
    5    & 29.55 & \textbf{28.24} \\
    6    & 29.55 & \underline{28.78}  \\
    7    & \textbf{23.70} & 49.01 \\
    8    & 35.27 & 51.72 \\
    9    & \underline{24.13} & 31.35 \\
    10   & 25.87 & 43.93 \\
    \hline
    \end{tabular}%
    \caption{Attack Success Rate (ASR, \%) under varying numbers of clusters $m$, the injection rate is 10\%. CG denotes ConceptGuard. Bold \textbf{values} highlight the best performance, while underlined \underline{values} indicate competitive performance. $m = 1$ refers to the ASR when ConceptGuard is not applied. }
  \label{tab: num_clus:ir=0.1}%
\end{table}%

\subsection{Injection Rate and Trigger Size}

Table \ref{tab:ij_ts} demonstrates the performance of ConceptGuard under varying injection rates and trigger sizes. At a 2\% injection rate, ConceptGuard significantly reduces the ASR while maintaining or slightly improving the ACC. For example, when the trigger size is 12, the ASR for CAT drops from 13.97\% to 12.3\%, and the ACC improves from 80.72\% to 81.77\%. Similarly, for CAT+, the ASR decreases from 38.88\% to 21.69\%, and the ACC increases from 80.46\% to 82.34\%.

At a 10\% injection rate, the attack success rate generally increases, but ConceptGuard still effectively reduces the ASR. For instance, with a trigger size of 12, the ASR for CAT drops from 38.08\% to 27.55\%, and the ACC improves from 74.66\% to 75.22\%. For CAT+, the ASR decreases from 68.49\% to 44.69\%, and the ACC increases from 75.11\% to 76.35\%.

As the trigger size increases, the effectiveness of ConceptGuard remains robust. For smaller trigger sizes (12, 15), ConceptGuard significantly reduces the ASR and maintains high ACC. For example, with a trigger size of 12, the ASR for CAT drops from 13.97\% to 12.3\%, and the ACC improves from 80.72\% to 81.77\%. For CAT+, the ASR decreases from 38.88\% to 21.69\%, and the ACC increases from 80.46\% to 82.34\%.

For larger trigger sizes (17, 20, 23), ConceptGuard continues to perform well, although the ASR increases. For instance, with a trigger size of 20, the ASR for CAT drops from 44.66\% to 11.55\%, and the ACC improves from 78.01\% to 78.75\%. For CAT+, the ASR decreases from 89.68\% to 17.16\%, and the ACC increases from 78.86\% to 78.56\%.

ConceptGuard effectively reduces the ASR and maintains or improves the ACC under various injection rates and trigger sizes. This robust performance highlights the effectiveness of ConceptGuard in protecting models against concept-level backdoor attacks, thereby enhancing the security and trustworthiness of the models.

\subsection{Target Class}
%% 表格 sheet 11
%% 分析

Table \ref{app:tc} demonstrates the performance of ConceptGuard across various target classes under a fixed injection rate of 5\% and a trigger size of 20, using the CUB dataset. For the unprotected models (CAT/CAT+), ASR varies significantly across different target classes. For example, for target class 8, the ASR for CAT is 74.24\%, which is substantially reduced to 21.93\% with ConceptGuard (CAT (CG)). Similarly, for target class 48, the ASR for CAT+ is 95.11\%, which is reduced to 22.14\% with ConceptGuard (CAT+ (CG)).

Across all target classes, ConceptGuard consistently improves ACC while significantly reducing ASR. For instance, for target class 16, the ACC for CAT increases from 75.16\% to 80.36\% with ConceptGuard, and the ASR drops from 40.91\% to 8.74\%. For target class 72, the ACC for CAT+ increases from 75.34\% to 80.62\% with ConceptGuard, and the ASR drops from 59.06\% to 10.11\%.

These results highlight the robustness of ConceptGuard in defending against concept-level backdoor attacks across different target classes. Despite variations in the target classes, ConceptGuard maintains its effectiveness in reducing the ASR and improving the ACC, thereby enhancing the overall security and reliability of the models. This consistent performance underscores the practical value of ConceptGuard in real-world applications where diverse and targeted attacks are a significant concern.

\subsection{The Impact of Number of Clusters in Other Attack Setting }
%% 表格 sheet 5-10 10%的
%% 分析

Table \ref{tab: num_clus:ir=0.1} illustrates the attack success rates (ASR, \%) under varying cluster numbers $m$ for CG(CAT) and CG(CAT+), with a 10\% injection rate. Both methods significantly reduce ASR compared to the experiment without defense ($m=1$), with CG(CAT) showing consistent improvement as $m$ increases and CG(CAT+) achieving optimal performance at moderate cluster numbers.

\end{document}